\documentclass{arxiv}
\usepackage{amssymb}
\usepackage{tikz}
\usetikzlibrary{arrows,calc,topaths}
\tikzstyle{small}=[font=\footnotesize]
\tikzset{
    every picture/.style={>=stealth,auto,node distance=5cm},
    label/.style={font=\scriptsize},
    mgnode/.style={font=\boldmath},
}
\usepackage{centernot}
\usepackage{color}
\usepackage{stmaryrd}

\newcommand{\Act}{\mathit{Act}}
\newcommand{\ignore}[1]{}

\newcommand{\PSPACE}{\textsc{PSpace}}

\newcommand{\N}{\mathbb{N}}

\newcommand{\Z}{\mathbb{Z}}

\newcommand{\x}{\times}

\newcommand{\G}{\mathcal{G}}

\newcommand{\mbsim}{\sim}

\newcommand{\mwbsim}{\approx}

\newcommand{\norm}[1]{\|#1\|}
\newcommand{\norminf}[1]{\norm{#1}_{\infty}}

\newcommand{\RCTL}{\text{RCTL}}
\newcommand{\CTL}{\text{CTL}}
\newcommand{\EF}{\text{EF}}
\newcommand{\EG}{\text{EG}}
\newcommand{\OP}[1]{\textit{#1}}
\newcommand{\CTLS}{\text{CTL*}}

\newcommand{\EGCCTLS}{\text{EGCCTL*}}
\newcommand{\AGCCTLS}{\text{AGCCTL*}}

\newcommand{\LTL}{\text{LTL}}

\newcommand{\Con}[1]{\mathcal{#1}}
\newcommand{\Pre}[2]{Pre_{#1}(#2)}

\newcommand{\DEN}[1]{\llbracket#1\rrbracket}
\newcommand{\exnext}[1]{\langle#1\rangle}

\newcommand{\Var}{{\it Var}}
\newcommand{\Val}{{\it Val}}
\newcommand{\Const}{{\it Const}}
\newcommand{\Sat}{{\it Sat}}
\newcommand{\Proj}{{\it Proj}}

\usepackage{ifthen}
\usepackage{centernot}
\usepackage{mathtools}

\newlength{\minlen}
\newlength{\minwlen}
\newlength{\arrowlen}
\newlength{\inputlen}
\newcommand{\Xrightarrow}[4][]{
  \ifthenelse { \equal {#1} {} }
    { }   
    { \settowidth{\minlen}{$\xrightarrow{#1}$} }  
  \settowidth{\inputlen}{$\xrightarrow{#2}$}
  \ifthenelse{\lengthtest{\the\minlen>\the\inputlen}}
           {\setlength{\arrowlen}{\minlen}}
           {\setlength{\arrowlen}{\inputlen}}
  \mathrel{\xrightarrow{\mathmakebox[\arrowlen]{#2}}\!\!{}^{\scriptstyle{#3}}_{\scriptstyle{#4}}}
  }
\newcommand{\XRightarrow}[4][]{
  \ifthenelse { \equal {#1} {} }
    { }   
    { \settowidth{\minwlen}{$\xRightarrow{#1}$} }  
  \settowidth{\inputlen}{$\xRightarrow{#2}$}
  \ifthenelse{\lengthtest{\the\minwlen>\the\inputlen}}
           {\setlength{\arrowlen}{\minwlen}}
           {\setlength{\arrowlen}{\inputlen}}
  \mathrel{\xRightarrow{\mathmakebox[\arrowlen]{#2}}\!\!{}^{\scriptstyle{#3}}_{\scriptstyle{#4}}}
  }

\newcommand{\notXrightarrow}[4][]{
  \ifthenelse { \equal {#1} {} }
    { }   
    { \settowidth{\minlen}{$\xRightarrow{#1}$} }  
  \settowidth{\inputlen}{$\xrightarrow{#2}$}
  \ifthenelse{\lengthtest{\the\minlen>\the\inputlen}}
           {\setlength{\arrowlen}{\minlen}}
           {\setlength{\arrowlen}{\inputlen}}
  \mathrel{\centernot{\xrightarrow{\mathmakebox[\arrowlen]{#2}}}\!\!{}^{\scriptstyle{#3}}_{\scriptstyle{#4}}}
  }
\newcommand{\notXRightarrow}[4][]{
  \ifthenelse { \equal {#1} {} }
    { }   
    { \settowidth{\minwlen}{$\xRightarrow{#1}$} }  
  \settowidth{\inputlen}{$\xRightarrow{#2}$}
  \ifthenelse{\lengthtest{\the\minwlen>\the\inputlen}}
           {\setlength{\arrowlen}{\minwlen}}
           {\setlength{\arrowlen}{\inputlen}}
  \mathrel{\centernot{\xRightarrow{\mathmakebox[\arrowlen]{#2}}}\!\!{}^{\scriptstyle{#3}}_{\scriptstyle{#4}}}
  }
\newcommand{\Step}[4][]{\Xrightarrow[#1]{#2}{#3}{#4}}

\newcommand{\step}[2][]{\Step[#1]{#2}{}{}}

\newcommand{\WStep}[4][]{\XRightarrow[#1]{#2}{#3}{#4}}

\newcommand{\wstep}[2][]{\WStep[#1]{#2}{}{}}

\usepackage{hyperref}
\newcommand{\qed}{}
\usepackage{todonotes}

\title{Branching-Time Model Checking Gap-Order Constraint Systems (Extended
    Version)}
\author{Richard Mayr\\ University of Edinburgh, UK
        \and
        Patrick Totzke\\ University of Warwick, UK
        }
\runninghead{R.~Mayr, P.~Totzke }{ Branching-Time Model Checking Gap-Order
    Constraint Systems}

\begin{document}
    \maketitle

    \begin{abstract}
We consider the model checking problem for Gap-order Constraint 
Systems (GCS) w.r.t.\ the branching-time temporal logic CTL, 
and in particular its fragments \EG\ and \EF.
GCS are nondeterministic infinitely branching processes 
described by evolutions of
integer-valued variables, subject to Presburger constraints
of the form $x-y\ge k$, where $x$ and $y$ are variables or constants 
and $k\in\N$ is a non-negative constant.
We show that \EG\ model checking is undecidable for GCS, while
\EF\ is decidable. In particular, this implies the 
decidability of strong and weak bisimulation equivalence between GCS 
and finite-state systems.

    \end{abstract}
   
    \section{Introduction}
Counter machines \cite{Min1967} extend a finite control-structure
with unbounded memory in the form of counters that can hold 
arbitrarily large integers (or natural numbers), and thus resemble basic programming
languages.
However, almost all behavioral properties, e.g., reachability and termination,
are undecidable for counter machines with two or more counters \cite{Min1967}.
For the purpose of formal software verification, various formalisms have been
defined that approximate counter machines and still retain the decidability
of some properties. E.g., Petri nets model weaker counters that cannot be 
tested for zero, and have a decidable reachability problem \cite{May1984}.

Gap-order constraint systems \cite{Rev1993,FR1996,Boz2012,BP2012} are another model that
approximates the behavior of counter machines. 
They are nondeterministic infinitely branching processes 
described by evolutions of
integer-valued variables, subject to Presburger constraints
of the form $x-y\ge k$, where $x$ and $y$ are variables or constants 
and $k\in\N$ is a non-negative constant.
Unlike in Petri nets,
the counters can be tested for zero, but computation steps still have a certain
type of monotonicity that yields a decidable reachability problem.
In fact, control-state reachability is decidable even for the more
general class of constrained multiset rewriting systems \cite{AD2006}.

\paragraph{Previous work.}
Beyond reachability, several model checking problems have been 
studied for GCS and related formalisms.
The paper \cite{Cer1994} studies Integral Relational Automata (IRA),
a model that is subsumed by 
GCS, that allows only constraints of the form $x\ge y$ or $x > y$,
where $y$ and $x$ are variables or constants.
It is shown that \CTL\ model checking of IRA is undecidable,
even for the restriction \RCTL, that forbids next-state modalities.
In contrast, model checking IRA remains decidable for
the existential and universal fragments of \CTLS.
Models of equal expressivity include the
monotonicity constraint systems (MCS) of \cite{Ben2009}
and $(\Z,<,=)$-automata in \cite{DD2007}.
Demri and D'Souza \cite{DD2007} show that satisfiability
and model checking 
\LTL\ is decidable and PSPACE-complete.

Bozzelli and Pinchinat \cite{Boz2012,BP2012}
study the more general model of
gap-order constraint systems (GCS), which
strictly extend the models mentioned above.
They show that model checking GCS is decidable and PSPACE-complete for the logic \EGCCTLS,
but undecidable for \AGCCTLS, which are the existential and universal fragments
of \CTLS, respectively, extended with gap constraints as atomic propositions.
Moreover, satisfiability is PSPACE-complete for both these fragments.
\EGCCTLS\ and  \AGCCTLS\ are not dual, 
since gap constraints are not closed under negation. 
Moreover, they are orthogonal to the fragments \EF\ and \EG\ considered in this paper, 
which allow nesting of negation and the operator $\OP{EF}$ (resp. $\OP{EG}$).
Checking fairness
(the existence of infinite runs
where a variable has a fixed value infinitely often)
and thus termination (the non-existence of infinite runs),
and also strong termination (the existence of a bound on the length
of all runs)
are decidable in polynomial space \cite{BP2012,Boz2012}.
An important ingredient for these results
are effectively constructible under-approximations
of the set of GCS runs induced by a given
sequence of transitions,
which preserve enabledness (Thm.~2 in \cite{BP2012}).
This comes at the cost of losing information
about the induced runs.
In particular,
it is impossible to recover (a representation of) the exact set of runs induced
by a sequence of transitions from its approximation.

\paragraph{Our contribution.}
We study the decidability of model checking problems for GCS
with fragments of computation-tree logic (CTL), namely \EG\ and \EF\
(see e.g.~\cite{Esp1997}).

We first show that \EG-model checking is undecidable, even for the weaker model of IRA \cite{Cer1994}.
On the other hand, model checking GCS with respect to
\EF\ remains decidable.
This positive result is based on the observation that
one can use boolean combinations of
gap constraints to represent the sets of variable valuations satisfying a given
\EF\ formula, and that
such a representation
can be effectively computed
in a bottom-up fashion.
An immediate consequence 
is that
checking strong and weak bisimulation equivalence is decidable between GCS and finite-state systems.

    \section{Gap-Order Constraint Systems}
Let $\Z$ and $\N$ denote the sets of integers and non-negative integers.
A \emph{labeled transition system} (LTS) is described by a triple
$T=(V,\Act,\step{})$ where $V$ is a (possibly infinite) set of
states, $\Act$ is a finite set of action labels and $\longrightarrow\,\subseteq V\x
\Act\x V$ is the labeled transition relation.
We 
use the infix notation $s\step{a}s'$ for a
transition $(s,a,s')\in\ \step{}$, in which case we say $T$ makes an \emph{$a$-step} from $s$ to
$s'$. 
For a set $S\subseteq V$ of states and $a\in\Act$ we define the set of $a$-predecessors by
$\Pre{a}{S} = \{s'|s'\step{a}s\in S\}$.
We write $\step{}^*$ for the transitive and reflexive closure of $\step{}$ and
let $Pre^*(S) = \{s'|s'\step{}^*s\in S\}$.

We fix a finite set $\Var$ of \emph{variables} ranging over the integers 
and a finite set $\Const\subseteq\Z$ of constants.
Let $\Val$ denote the set of variable \emph{valuations} $\nu:\Var\to\Z$.
To simplify the notation, we will extend the domain of valuations 
to constants, where they behave as the identity, 
i.e., $\nu(c)=c$ for all $c\in \Z$.

\begin{definition}[Gap Constraints]
    A \emph{gap clause} over ($\Var,\Const$)
    is an inequation of the form
    \begin{equation}
      (x-y\ge k)
    \end{equation}
    where $x,y\in \Var\cup \Const$ and $k\in \Z$.
    A clause is called \emph{positive} if $k\in\N$.
    A (positive) gap \emph{constraint} is a finite conjunction of (positive) gap clauses.
    A \emph{gap formula} is an arbitrary boolean combination of gap clauses.

    A valuation $\nu:\Var\to\Z$ satisfies the clause $\Con{C}:(x-y)\ge k$
    (write $\nu\models \Con{C}$) if it respects the prescribed inequality. That is,
    \begin{equation}
        \nu\models (x-y)\ge k \iff \nu(x)-\nu(y)\ge k.
    \end{equation}
    We define the satisfiability of arbitrary gap formulae inductively in the usual 
    fashion and write
    $\Sat(\varphi) = \{\nu\in \Val\ |\ \nu\models\varphi\}$
    for the set of valuations that satisfy the formula $\varphi$. In particular, a valuation
    satisfies a gap constraint iff it satisfies all its clauses.
    A set $S\subseteq \Val$ of valuations is called \emph{gap definable} if there is a gap formula
    $\varphi$ with $S=\Sat(\varphi)$.
\end{definition}

We will consider processes whose states are described by valuations and whose
dynamics is described by stepwise changes in these variable valuations,
according to positive gap constraints.

Let $\Var'=\{x' \,|\, x\in \Var\}$ be the set of primed copies of the variables.
These new variables are used to express constraints on how values can change when moving from one
valuation to another: $x'$ is interpreted as the next value of variable $x$.
A \emph{transitional} gap clause (-constraint, -formula) is a gap clause
(-constraint, -formula) with variables in $\Var\cup \Var'$.
The combination $\nu\oplus\nu':\Var\cup \Var'\to\Z$
of two valuations $\nu,\nu':\Var\to \Z$ maps
variables $x\in\Var$ to $\nu\oplus\nu'(x)=\nu(x)$
and $x'\in\Var'$ to $\nu\oplus\nu'(x')=\nu'(x)$.

Transitional gap clauses can be used as conditions on how valuations may evolve in one step.
For instance, $\nu$ may change to $\nu'$ only if $\nu\oplus\nu'\models\varphi$ for some
gap clause $\varphi$.

\begin{definition}
    A \emph{Gap-Order Constraint System} (GCS) is given by 
    a finite set of \emph{positive} transitional gap constraints
    together with a labeling function.
    Formally, a GCS is a tuple $\G=(\Var,\Const,\Act,\Delta,\lambda)$ where $\Var,\Const,\Act$ are
    finite sets of variables, constants and action symbols, $\Delta$ is a finite set of
    \emph{positive} transitional gap constraints over $(\Var,\Const)$ and
    $\lambda:\Delta\to\Act$ is a labeling function. 
    Its operational semantics is given by an infinite LTS with states $\Val$ where
    \begin{equation}
        \nu\step{a}\nu' \iff \nu\oplus\nu'\models\Con{C}
    \end{equation}
    for some constraint $\Con{C}\in\Delta$ with $\lambda(\Con{C})=a$.
    For a set $M\subseteq \Val$ of valuations we write $\Pre{\Con{C}}{M}$ for the set
    $\{\nu\ |\ \exists \nu'\in M.\, \nu\oplus\nu'\models\Con{C}\}$
    of $\Con{C}$-predecessors.
\end{definition}

Observe that a positive gap constraint $(x-0 \ge 0)\ \land\ (0-x\ge 0)$ is satisfied only by valuations
assigning value $0$ to variable $x$.
Similarly, one can test if a valuation equates two variables.
Also, it is easy to simulate a finite control in a GCS using additional variables.\footnote{
    In fact, \cite{BP2012,Boz2012} consider an equivalent notion of GCS that
    explicitly includes a finite control.}
What makes this model computationally non-universal is the fact that we demand \emph{positive} constraints:
while one can easily demand an increase or decrease of variable $x$ by \emph{at least} some offset $k\in\N$,
one cannot demand a difference of \emph{at most} $k$ (nor exactly $k$).

\begin{example}
  \label{ex:countdown}
Consider the GCS with variables $\{x,y\}$ and single constant $\{0\}$ with two constraints
$\Delta=\{\Con{C}X,\Con{C}Y\}$ for which
$\lambda(\Con{C}X)=a$ and $\lambda(\Con{C}Y)=b$.
\begin{align}
    \Con{C}X = &((x-x'\ge 1) \land\ (y'-y\ge 0)\ \land\ (y-y'\ge 0)\ \land\ (x'-0\ge 0))\\
    \Con{C}Y = &((y-y'\ge 1)\ \land\ (x'-x\ge 0)\ \land\ (y'-0\ge 0)).
\end{align}
This implements a sort of lossy countdown where every step
strictly decreases the tuple $(y,x)$ lexicographically:
$\Con{C}X$ induces $a$-steps that decrease $x$ while preserving the value of $y$
and $\Con{C}Y$ induces $b$-steps that increase $x$ arbitrarily but have to decrease $y$ at the same time.
The last clauses in both constraints ensure that $x$ and $y$ never
change from a non-negative to a negative value.

\end{example}
In the sequel, we allow ourselves to abbreviate constraints for the sake of readability.
For instance, the constraint $\Con{C}X$ in the previous example could equivalently be written as
$(x>x'\ge 0)\ \land\ (y=y')$.

    \section{Branching-Time Logics for GCS}
We consider (sublogics of) the branching-time logic CTL over processes defined by gap-order
constraint systems, where atomic propositions are gap clauses.
The denotation of an atomic proposition $\mathcal{C} = (x-y\ge k)$ is $\DEN{\mathcal{C}} = Sat(\mathcal{C})$,
the set of valuations satisfying the constraint.
Well-formed CTL formulae are inductively defined by the following grammar, where $\Con{C}$ ranges
over the atomic propositions and $a \in \Act$ over the action symbols.
\begin{equation}
    \label{eq:grammar}
    \psi ::= \Con{C} \ \  |\ \  true\ \  |\ \  \lnot \psi\ \  |\ \  \psi \land \psi\ \  |\ \  \exnext{a}\psi
    \ \  |\ \  \OP{EF}\psi\ \  |\ \  \OP{EG}\psi\ \  |\ \  \OP{E}(\psi \OP{U} \psi)
\end{equation}
To define the semantics, we fix a GCS $\G$. 
Let ${\it Paths}^\omega(\nu_0)$ be the set of infinite derivations
\begin{equation}
    \pi = \nu_0\step{a_0}\nu_1\step{a_1}\nu_2\dots
\end{equation}
of $\G$ starting with valuation $\nu_0\in \Val$ and let $\pi(i)=\nu_i$ denote the $i$-th valuation $\nu_i$
on $\pi$. Similarly, we write ${\it Paths}^*(\nu_0)$ for the set of finite
derivations starting from $\nu_0$.
The denotation of formulae, with respect to the fixed GCS $\G$, is defined in the standard way.
\begin{align}
    \DEN{\Con{C}} &= Sat(\Con{C})\\
    \DEN{{\it true}} &= \Val\\
    \DEN{\lnot \psi} &= \Val\setminus \DEN{\psi}\\
    \DEN{\psi_1\land\psi_2} &= \DEN{\psi_1}\cap\DEN{\psi_2}\\
    \DEN{\exnext{a}\psi} &= \Pre{a}{\DEN{\psi}}\\
    \DEN{\OP{EF}\psi} &= \{\nu\ |\ \exists \pi \in Paths^*(\nu).\ \exists i\in\N.\ \pi(i)\in\DEN{\psi}\}\\
    \DEN{\OP{EG}\psi} &= \{\nu\ |\ \exists \pi \in Paths^\omega(\nu).\ \forall i\in\N.\ \pi(i)\in\DEN{\psi}\}\\
    \DEN{\OP{E}(\psi_1 \OP{U} \psi_2)} &= \{\nu\ |\ \exists \pi\in Paths^*(\nu).\ \exists i\in\N.
        \pi(i)\in\DEN{\psi_2} \land
        \forall j<i. \pi(j)\in\DEN{\psi_1}\}
\end{align} 
We use the usual syntactic abbreviations
${\it false} = \lnot {\it true}$,
$\psi_1\lor\psi_2 = \lnot(\lnot\psi_1\land\lnot\psi_2)$.

The sublogics $\EF$ and $\EG$ are defined by restricting the grammar \eqref{eq:grammar} defining
well-formed formulae: \EG\ disallows subformulae of the form $\OP{E}(\psi_1 \OP{U}\psi_2)$ 
and $\OP{EF}\psi$ and
in \EF, no subformulae of the form $\OP{E}(\psi_1 \OP{U} \psi_2)$ or $\OP{EG}\psi$ are allowed.
The \emph{Model Checking Problem} is the following decision problem.

\vspace{0.3cm}
\begin{tabular}{ll}
  {\sc Input:}  &A GCS $G=(\Var,Const,\Act,\Delta,\lambda)$, a valuation $\nu:\Var\to\Z$\\
              &and a formula $\psi$.\\
  {\sc Question:} &$\nu\models \psi$?
\end{tabular}
\vspace{0.3cm}

Cerans \cite{Cer1994} showed that general \CTL\ model checking is undecidable
for gap-order constraint systems.
This result holds even for restricted \CTL\ without \emph{next} operators $\exnext{a}$.
In the following section, we show a similar undecidability result for the fragment \EG.  
On the other hand, model checking GCS with the fragment \EF\ turns out to be decidable;
cf.\ Section~\ref{sec:EF}.

    \section{Undecidability of EG Model Checking}
    \label{sec:EG}
\begin{theorem}
\label{thm:eg}
    The model checking problem for \EG\ formulae over GCS is undecidable.
\end{theorem}

\begin{proof}
    By reduction from the halting problem of deterministic $2$-counter Minsky Machines (2CM).
    $2$-counter machines consist of a deterministic finite control, including a designated halting
    control-state ${\it halt}$ and two integer counters that can be incremented and decremented 
    by one and tested
    for zero. Checking if such a machine reaches the halting state 
    from an initial configuration with control-state ${\it init}$ and 
    counter values $x_1=x_2=0$ is undecidable \cite{Min1967}.
    Given a 2CM $M$, we will construct a GCS together with an initial valuation $\nu_0$ and a \EG\
    formula $\psi$ such that $\nu_0\models\psi$ iff $M$ does not halt.

    First of all, observe that we can simulate a finite control of $n$ states using one additional
    variable $state$ that will only ever be assigned values from $1$ to $n$.
    To do this, let $[p]\le n$ be the index of state $p$ in an arbitrary enumeration of the state
    set. Now, a transition $p\step{}q$ from state $p$ to $q$ introduces the constraint
    $(state=[p]\land state' = [q])$. We will abbreviate such constraints by $(p\step{}q)$ in the
    sequel and simply write $p$ to mean the clause $(state=[p])$.

    We use two variables $x_1,x_2$ to act as integer counters. Zero-tests can then directly be
    implemented as constraints $(x_1=0)$ or $(x_2=0)$.
    It remains to show how to simulate increments and decrements by exactly $1$.
    Our GCS will use two auxiliary variables $y,z$ and a new state ${\it err}$.
    We show how to implement increments by one; decrements can be done analogously.

    Consider the $x_1$-increment $p\step{x_1=x_1+1}q$ that takes the 2CM from state $p$ to $q$ and
    increments the counter $x_1$. The GCS will simulate this in two steps, as depicted in
    Figure~\ref{fig:eg-inc} below.
    \begin{figure}[ht]
           \begin{center}
            \begin{tikzpicture}
                \node(S) at (0,0) {$p$};
                \node(M) at (5,0) {$to_q$};
                \node(S') at (10,0) {$q$};
                \node(H) at (5,-1.5) {${\it err}$};
                \path[->] (S) edge node {$x_1'>x_1=y'$} (M);
                \path[->] (M) edge node {$y<z'<x_1$} (H);
                \path[->] (M) edge node {$x_1'=x_1$} (S');
           \end{tikzpicture}
           \end{center}
    \caption{Forcing faithful simulation of $x_1$-increment.
        All steps contain the additional constraint $x_2'=x_2$, which is not shown,
        to preserve the value of the other counter $x_2$.
    \label{fig:eg-inc}
    }
    \end{figure}
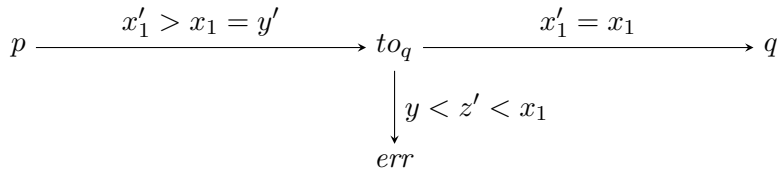
    The first step can arbitrarily increment $x_1$ and will remember (in variable $y$) the old value
    of $x_1$. The second step does not change any values and just moves to the new control-state.
    However, incrementing by more than one in the first step enables an extra move to the error state
    ${\it err}$ afterwards. This error-move is enabled if one can assign a value to variable $z$ that is
    strictly in between the old and new value of $x_1$, which is true iff the increment in step 1
    was not faithful.
    The incrementing transition of the 2CM is thus translated to the following three constraints.
    \begin{align}
        &(p\step{}to_q)\land (x_1'>x_1=y') \land (x_2'=x_2)\\
        &(to_q\step{}q)\land (x_1'=x_1)\land (x_2'=x_2)\\
        &(to_q\step{}{\it err})\land (y<z'<x_1).
    \end{align}

    If we translate all operations of the 2CM into the GCS formalism as indicated above,
    we end up with an over-approximation of the 2CM that allows runs that faithfully
    simulate runs in the 2CM but also runs which `cheat' and possibly
    increment or decrement by more than one and still don't go to state ${\it err}$ in the following step.

    We enforce a faithful simulation of the 2CM by using the formula that is to be checked, demanding
    that the error-detecting move is never enabled.
    The GCS will only use a unary alphabet $\Act=\{a\}$ to label constraints. In particular,
    observe that the formula $\exnext{a}err$ holds in every configuration which can move to state
    ${\it err}$ in one step. Now, the \EG\ formula 
    \begin{equation}
        \psi= \OP{EG}(\lnot halt\ \land\ \lnot\, \exnext{a}{\it err})
    \end{equation}
    asserts that there is an infinite path which never visits state ${\it halt}$ and
    along which no step to state ${\it err}$ is ever enabled.
    This means $\psi$ is satisfied by valuation $\nu_0=\{state=[init], x_1=x_2=y=z=0\}$ iff there is
    a faithful simulation of the 2CM from initial state ${\it init}$ with both counters set to $0$
    that never visits the halting state.
    Since the 2CM is deterministic, there is only one way
    to faithfully simulate it and hence $\nu_0\models\psi$ iff the 2CM does not halt.
    Notice that the constructed GCS is in fact an IRA \cite{Cer1994},
    since it only uses gap constraints of the form $x>y$ or $x=y$.
    \qed
\end{proof}

    \section{Decidability of EF Model Checking}
    \label{sec:EF}
Let us fix sets $\Var$ and $\Const$ of variables and constants, respectively.
We will use an alternative characterization of gap constraints
in terms of
\emph{monotonicity graphs}
\footnote{These were called \emph{Graphose Inequality Systems} in \cite{Cer1994}
    and \emph{gap-graphs} in \cite{Rev1993}.
},
which are finite graphs with nodes $\Var\cup \Const$.
Monotonicity graphs are used to represent sets of variable valuations.
We show that so represented sets
are effectively closed under all logical connectors allowed in \EF,
and one can thus evaluate a formula bottom up.

\begin{definition}[Monotonicity Graphs]
    \label{def:MGs}
    A \emph{monotonicity graph} (MG) over $(\Var, \Const)$ is a finite, directed graph $M=(V,E)$ with
    nodes $V=\Var\cup \Const$ and in which each edge in $E$ carries a
    \emph{weight} in $\Z\cup\{-\infty,\infty\}$.
    The \emph{degree} of $M$ is the
    largest $k\in\N$ such that there is an edge with weight $-k$ in $M$
    or $0$ if no edge has weight in $\Z\setminus\N$.
    The degree of a set $\{M_0,M_1,\dots,M_j\}$ of MG is defined as the
    maximal degree of any MG $M_k$ in the set.

    A valuation $\nu:\Var\to\Z$ satisfies $M$
    (write $\nu\models M$) if for every edge
    $(x\step{k}y)$ it holds that
    $\nu(x)-\nu(y)\ge k$.
    Let $\Sat(M)$ denote the set of
    valuations satisfying $M$.
    A set $S\subseteq \Val$ is \emph{MG-definable} if there is a finite set
    $\{M_0,M_1,\dots,M_j\}$ of MG such that
    \begin{equation}
      S=\bigcup_{0\le i\le j} Sat(M_i)
    \end{equation}
    and called $\text{MG}^n$-definable if there is such a set of MG with degree
    $\le n$. We write $\textit{MG}$ and $\textit{MG}^n$ for the classes of
    $\text{MG}$- and $\text{MG}^n$-definable sets respectively.
    
    For a monotonicity graph $M$,
    we write $M(x,y) \in\{-\infty,\infty\}\cup\Z$ for the least upper bound
    of the cumulative weight of all paths from node $x$ to node $y$. 
    Note that this is $-\infty$ if there is no such path.
    The \emph{closure} $|M|$ is the unique complete MG with edges $x\step{M(x,y)}y$ for all $x,y \in \Var\cup \Const$.
\end{definition}
The following lemma and definition state some basic properties of monotonicity
graphs that can easily be verified; see \cite{Cer1994}.

\begin{lemma}
\label{lem:mg-basics}
  \
  \begin{enumerate}
      \item $\Sat(M)=\emptyset$ holds for any monotonicity graph $M$ that contains an edge
          with weight $\infty$
          or a cycle with positive weight sum.
      \item $|M|$ is polynomial-time computable from $M$ and $\Sat(M) = \Sat(|M|)$.
      \item If we fix sets $\Var,\Const$ of variables and constants then
          for any gap constraint $\Con{C}$ there is a unique monotonicity graph
          $M_{\Con{C}}$,
          containing an edge $x\step{k}y$ iff there is a clause $x-y\ge k$ in $\Con{C}$.
          Moreover, $\Sat(M_{\Con{C}})=\Sat(\Con{C})$.
  \end{enumerate}
\end{lemma}
The last point of this lemma states that monotonicity graphs and gap constraints are equivalent
formalisms. We call a MG \emph{positive} if it has degree $0$. Positive MG are
equivalent to positive gap constraints.
We thus talk about \emph{transitional} monotonicity graphs over $(\Var,\Const)$ as those
with nodes $\Var\cup \Var'\cup \Const$.
We further define the following operations on MG.
\begin{definition}
\label{def:mg-ops}
    Let $M,N$ be monotonicity graphs over $\Var, \Const$ and $V\subseteq \Var$.
    \begin{itemize}
        \item The \emph{restriction} $M|_V$ of $M$ to $V$ is the maximal subgraph of $M$ with nodes
              $V\cup \Const$.
        \item The \emph{projection} $\Proj(M,V) = |M|_V$ is the restriction of $M$'s closure to
              $V$.
        \item The \emph{intersection} $M\otimes N$ is the MG that contains an edge $x\step{k}y$ if
              $k$ is the maximal weight of any edge from $x$ to $y$ in $M$ or $N$.
        \item The \emph{composition} $G\circ M$ of a \emph{transitional} MG $G$ and $M$ is obtained
              by consistently renaming variables in $M$ to their primed copies, intersecting the
              result with $G$ and projecting to $\Var\cup \Const$.
              $G\circ M:= \Proj(M_{[\Var\mapsto \Var']}\otimes G, \Var)$.
    \end{itemize}
\end{definition}
These operations are surely computable in polynomial time.
The next lemma states important properties of these operations; see also \cite{Cer1994,BP2012}.
\begin{lemma}
\label{lem:mg-ops}
    \ 
    \begin{enumerate}
        \item $\Sat(\Proj(M,V)) = \{\nu|_V : \nu\in \Sat(M)\}$.
        \item $\Sat(M\otimes N) = \Sat(M)\cap \Sat(N)$
        \item $\Sat(G\circ M) = \{\nu\ |\ \exists \nu'\in \Sat(M).\ \nu\oplus\nu'\in \Sat(G)\} = \Pre{G}{M}$.
        \item If $M$ has degree $n$ and $G$ is a transitional MG of degree $0$,
            then $G\circ M$ has degree $\le n$.
    \end{enumerate}
\end{lemma}

\begin{example}
  \label{ex:mg}
The monotonicity graph on the left below corresponds to the contraint $\Con{C}X$ in
Example~\ref{ex:countdown}. On the right we see its closure (where edges with
weight $-\infty$
are omitted). Both have degree $0$.
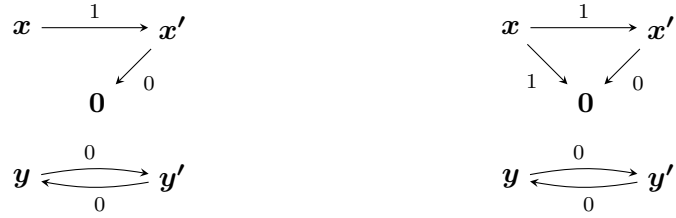
\begin{center}
  \begin{minipage}[c]{.4\textwidth}
    \centering
    \begin{tikzpicture}[scale=1.0,node distance=2cm]  
        \node(x) [mgnode] {$x$};
        \node(x')[mgnode,right of=x] {$x'$};
        \node(y) [mgnode,below of=x] {$y$};
        \node(y')[mgnode,right of=y] {$y'$};
        \node (0)[mgnode] at ($(x)!0.5!(y')$) {$0$};

        \path[->] (x) edge node[label] {$1$} (x');
        \path[->] (y') edge[bend left=10] node[auto,label,pos=0.45] {$0$} (y);
        \path[->] (y) edge[bend left=10] node[auto,label,pos=0.45] {$0$} (y');
        \path[->] (x') edge node[auto,label] {$0$} (0);
    \end{tikzpicture}
  \end{minipage}
  \begin{minipage}[c]{.4\textwidth}
    \centering
    \begin{tikzpicture}[scale=1.0,node distance=2cm]  
        \node(x) [mgnode] {$x$};
        \node(x')[mgnode,right of=x] {$x'$};
        \node(y) [mgnode,below of=x] {$y$};
        \node(y')[mgnode,right of=y] {$y'$};
        \node (0)[mgnode] at ($(x)!0.5!(y')$) {$0$};

        \path[->] (x) edge node[label] {$1$} (x');
        \path[->] (x) edge node[label,swap] {$1$} (0);
        \path[->] (y') edge[bend left=10] node[auto,label,pos=0.45] {$0$} (y);
        \path[->] (y) edge[bend left=10] node[auto,label,pos=0.45] {$0$} (y');
        \path[->] (x') edge node[auto,label] {$0$} (0);
    \end{tikzpicture}
  \end{minipage}
\end{center}
Let us compute the $\Con{C}X$-predecessors of the set $S= \{\nu\ |\ \nu(x)>\nu(y)=0\}$ which is
characterized by the single MG on the right below.
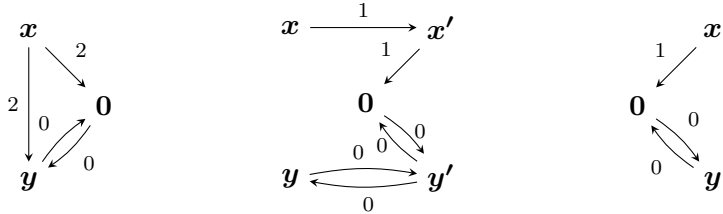
\begin{center}
  \begin{minipage}[c]{.25\textwidth}
    \centering
    \begin{tikzpicture}[scale=1.0,node distance=2cm] 
        \node (0)[mgnode] at (1,0) {$0$};
        \node(x) at (0,1) [mgnode] {$x$};
        \node(y) [mgnode] at (0,-1){$y$};
        \path[->] (0) edge[bend left=10] node[auto,label,pos=0.45] {$0$} (y);
        \path[->] (y) edge[bend left=10] node[auto,label,pos=0.45] {$0$} (0);
        \path[->] (x) edge node[auto,label] {$2$} (0);
        \path[->] (x) edge node[swap,label] {$2$} (y);
    \end{tikzpicture}
  \end{minipage}
  \begin{minipage}[c]{.25\textwidth}
    \centering
    \begin{tikzpicture}[scale=1.0,node distance=2cm]  
        \node (0)[mgnode] at (0,0) {$0$};
        \node(x) [mgnode] at (-1,1) {$x$};
        \node(x')[mgnode] at (1,1) {$x'$};
        \node(y) [mgnode] at (-1,-1) {$y$};
        \node(y')[mgnode] at (1,-1) {$y'$};
  
        \path[->] (x) edge node[label] {$1$} (x');
        \path[->] (y') edge[bend left=10] node[below,label,pos=0.45] {$0$} (y);
        \path[->] (y) edge[bend left=10] node[above,label,pos=0.45] {$0$} (y');
        \path[->] (x') edge node[swap,label] {$1$} (0);
  
        \path[->] (0) edge[bend left=10] node[right,label,pos=0.45] {$0$} (y');
        \path[->] (y') edge[bend left=10] node[left,label,pos=0.45] {$0$} (0);
    \end{tikzpicture}
  \end{minipage}
  \begin{minipage}[c]{.25\textwidth}
    \centering
    \begin{tikzpicture}[scale=1.0,node distance=2cm]  
        \node (0)[mgnode] at (0,0) {$0$};
        \node(x) at (1,1) [mgnode] {$x$};
        \node(y) [mgnode] at (1,-1){$y$};
        \path[->] (0) edge[bend left=10] node[auto,label,pos=0.45] {$0$} (y);
        \path[->] (y) edge[bend left=10] node[auto,label,pos=0.45] {$0$} (0);
        \path[->] (x) edge node[swap,label] {$1$} (0);
    \end{tikzpicture}
  \end{minipage}
\end{center}
If we rename variables $x$ and $y$ to $x'$ and $y'$ and intersect the result with $M_{\Con{C}X}$ we
get the MG in the middle. We project into $Var\cup Const$ by computing the closure and restricting
the result accordingly. This leaves us with the MG on the left, which characterizes the set
$\Pre{\Con{C}X}{S}=\{\nu\ |\ \nu(x)\ge 2\ \land\ \nu(y)=0\}$ as expected.

\end{example}

We have seen how to construct a representation of the $\Con{C}$-predecessors $\Pre{\Con{C}}{S}$ and
thus $\Pre{a}{S}$ for a MG-definable set $S$, gap constraint $\Con{C}$ and action $a\in\Act$.
The next lemma is a consequence of Lemma~\ref{lem:mg-basics}, point 3 and asserts that we can do the
same for complements.
\begin{lemma}
    \label{lem:neg}
    The class of MG-definable sets is effectively closed under complementation.
\end{lemma}
\begin{proof}
    By Lemma~\ref{lem:mg-basics} we can interpret a finite set of MG $\mathcal{M} = \{M_0,M_1,\dots,M_k\}$ as
    a gap formula in DNF. One can then use De Morgan's laws to propagate negations to atomic
    propositions, which are gap clauses of the form $x-y\ge k$. The negation is
    then expressible as $x-y < k$, which is equivalent to $y-x > -k$ and thus
    to the gap clause $y-x\ge -k+1$. It remains to bring the formula into DNF again, 
    which can then be described by finitely many MGs.
    \qed
\end{proof}
Observe that complementation potentially constructs MG with increased degree.
This next degree is bounded by the largest finite weight in the current graph minus one,
but nevertheless, an increase of degree cannot be avoided.
Therefore, classes of $\text{MG}^n$-definable sets are \emph{not} closed
under complementation.
\begin{example}
    \label{ex:3}

The set $S= \{\nu\ |\ \nu(x) - \nu(y) \ge 5\}$ corresponds to the gap-formula
$\varphi = (x-y\ge 5)$. Its MG $\{(x\step{5}y)\}$ is of degree $0$.
However, its complement is characterized by the MG $\{(y\step{-4}x)\}$,
which has degree $4$.

\end{example}

\newcommand{\mgle}{\sqsubseteq}
It remains to show that we can compute ${\it Pre}^*(S)$ for MG-definable sets $S$.
We recall the following partial ordering on monotonicity 
graphs and its properties \cite{Cer1994}.
\begin{definition}
    Let $M,N$ be MG over $(\Var, \Const)$. We say that 
    $M$ \emph{covers} $N$ (write $N\mgle M$) if for all
    $x,y\in \Var\cup \Const$ it holds that $N(x,y)\le M(x,y)$.
\end{definition}
\begin{lemma}\
\label{lem:mg-order}
    \begin{enumerate}
        \item If $N\mgle M$ then $Sat(N)\supseteq Sat(M)$.
        \item
            For every $n\in\N$,
            $\sqsubseteq$ is a well-order on the set 
            of MG over $(\Var,\Const)$ with degree $\le n$.
    \end{enumerate}
\end{lemma}
\begin{proof}
    For the first claim, assume $\nu\in Sat(M)=Sat(|M|)$.
    Then, for every $x,y\in\Var\cup\Const$, we have
    $\nu(x)-\nu(y)\ge M(x,y)\ge N(x,y)$. So
    $\nu\in Sat(|N|)=Sat(N)$.

    The second claim follows from Dickson's Lemma if we interpret each
    MG $M$ with degree $n$ as $|\Var\cup\Const|^2$-dimensional vector where
    the component for the pair $(x,y)$ has value $n+M(x,y)$.
\end{proof}
Note that point 1 states that a $\mgle$-bigger MG is more restrictive and hence has a smaller
denotation. Also notice that $\mgle$ is \emph{not} a well order on the set of all MG
due the lack of a bound on finite, negative weights:
for instance, the sequence $(M_n)_{n\in\N}$ of MG, where
for every $n$, the graph $M_n$
has edges $x\step{n}y\step{-n}x$, is an infinite antichain.

\begin{lemma}
\label{lem:ef}
Let $S$ be a $\text{MG}^n$-definable set of valuations.
Then ${\it Pre}^*(S)$ is $\text{MG}^n$-definable and a representation of ${\it Pre}^*(S)$
  can be computed from a representation of $S$.
\end{lemma}
\begin{proof}
    It suffices to show the claim for a set $S$
    characterized by a single monotonicity graph $M_S$, because
    ${\it Pre}^*(S\cup S') = {\it Pre}^*(S)\cup {\it Pre}^*(S')$. 
    Assume that $M_S$ has degree $n$.

    We proceed by exhaustively building a finite tree of MG, starting in $M_S$. For every node $N$ we
    compute children $G\circ N$ for all of the finitely many transitional MG $G$ in the system.
    Point 4) of Lemma~\ref{lem:mg-ops} guarantees that all intermediate representations have
    degree $\le n$.  By Lemma~\ref{lem:mg-order}, point 2, any branch eventually ends in a node that
    covers a previous one and Lemma~\ref{lem:mg-order}, point 1 allows us to stop exploring such a
    branch.
    We conclude that ${\it Pre}^*(M)$ can be characterized by the finite union of all intermediate MG.
    \qed
\end{proof}
Finally, we are ready to prove our main result.
\begin{theorem}
\label{thm:EF-dec}
    \EF\ model checking is decidable for Gap-order constraint systems.
    Moreover, the set $\DEN{\psi}$ of valuations satisfying an \EF-formula $\psi$ is effectively
    gap definable.
\end{theorem}
\begin{proof}
    We can evaluate a formula bottom up, representing the sets satisfying subformulae by finite sets
    of MG. Atomic propositions are either ${\it true}$ 
    or gap clauses and can thus be written directly as
    MG. For composite formulae we use the properties that MG definable 
    sets are effectively closed under
    intersection (Lemma~\ref{lem:mg-ops}) and negation (Lemma~\ref{lem:neg}),
    and that we can compute
    representations of $\Pre{a}{S}$ and ${\it Pre}^*(S)$ for MG-definable sets $S$ by
    Lemmas~\ref{lem:mg-ops} and \ref{lem:ef}.

    The key observation is that although negation (i.e., complementing) 
    may increase the degree of the intermediate
    MG, this happens only finitely often in the bottom up evaluation of an \EF\ formula. 
    Computing representations for
    modalities $\langle a\rangle$ and $\OP{EF}$ does not increase the degree.
    \qed
\end{proof}
\begin{remark}
    Since steps in GCS are described by positive transitional gap constraints,
    it is straightforward to extend
    this positive result of Theorem~\ref{thm:EF-dec} to
    model checking GCS w.r.t.~the slightly more general logic
    $\EF_\mathcal{C}$, in which the next-state and reachability
    modalities
    $\exnext{a}_\Con{C}$ and
    $\OP{EF}_\Con{C}$ are subject to
    transitional gap clauses $\Con{C}$.
\end{remark}

The exact complexity of the EF model checking problem for GCS is still open. However,
a \PSPACE\ lower bound already holds for reachability of the simpler model of
{\em boolean programs}.
Moreover, a rather crude Ackermannian upper complexity bound on EF model
checking for GCS can be obtained by bounding
``bad sequences'' in our use of Dickson's Lemma.
In the remainder of this section we will elaborate on these reductions.

\subsection{A \PSPACE\ lower complexity bound}
Several equivalent notions of boolean programs are used in
different application domains (see e.g., \cite{CS1999,BR2000}).
Essentially, they consist of a
finite control unit that manipulates finitely many boolean variables.

We define \emph{boolean programs} as finite-state machines
with transitions of the form $s\step{g(\vec{x})/a}t$,
where $s$ and $t$ are control-states,
$a$ is an assignment $x = 0$ or $x=1$ for some variable $x$, and
$g(\vec{x})$ is a boolean formula with free variables $\vec{x}$.
The semantics of boolean programs is given by the binary step relation
$\step{}$ over pairs of control-states and variable valuations.
Let $\Var=\{x_1,x_2,\dots,x_k\}$ be the set of variables in the system
and let $\nu,\nu':\Var\to\{0,1\}$ be two valuations.
A transition $s\step{g(\vec{x})/y=b}t$ induces a step
$s,\nu \step{} t,\nu'$ if
1) $\nu\models g$, 2) $\nu'(y) = b$
and 3) $\nu'(x) = \nu(x)$ for $x\in\Var\setminus\{y\}$.

The \emph{state-reachability problem} for boolean programs asks,
for two given control states $s$ and $t$,
whether there exists a valuation $\nu$ and a
finite number of steps from $s,\nu_0$ to $t,\nu$.
Here, $\nu_0: x\mapsto 0$ assigns the value $0$ to every variable.
We will show that this problem is \PSPACE\ hard,
by reduction from the \PSPACE-complete satisfiability problem for
\emph{quantified boolean formulae} (QBF).
Notice that boolean programs can be directly simulated by gap-order constraint
systems. The same lower bound thus holds for the reachability problem, and
consequently also for the EF model checking problem for GCS.

\newcommand{\Qin}[1]{\mathit{eval}_{#1}}
\newcommand{\Qout}[1]{\mathit{out}_{#1}}
Let $Q_1x_1 Q_2x_2 \dots Q_kx_k \varphi$ be a QBF formula in prenex normal
form.
We construct a boolean program that contains control-states $\Qin{i}$ and $\Qout{i}$
as well as variables $x_i$ for each $1\le i \le k$.
The program evaluates the formula top down:
a subformula $\varphi_i = Q_ix_i, Q_{i+1}x_{i+1}\dots Q_kx_k \varphi$ is verified by a run
from control-state $\Qin{i}$ to $\Qout{i}$.
If the current valuation does not satisfy the subformula, the program deadlocks
and $\Qout{i}$ is not reachable.

The quantifier-free subformula $\varphi$
is directly evaluated using a transition $\Qin{k+1}\step{\varphi/y_{k+1}=1} \Qout{k+1}$.
Notice that a pair $\Qin{i},\nu$ is a deadlock unless $\nu\models\varphi$.
For an existential quantifier $Q_i$, there are transitions
\begin{equation*}
\Qin{i}\step{/x_i=0} \Qin{i+1},
\quad
\Qin{i}\step{/x_i=1} \Qin{i+1},
\quad
\Qout{i+1}\step{} \Qout{i}
\end{equation*}
For a universal quantifier $Q_i$, there is an extra variable $y_i$ and transitions
\begin{align*}
&\Qin{i}\step{/x_i=0} \Qin{i+1},
&&\Qin{i}\step{y_{i}=1/x_i=1} \Qin{i+1},\\
&\Qout{i+1}\step{/y_{i}=1} \Qin{i},
&&\Qout{i+1}\step{y_{i}=1/y_{i}=0} \Qout{i}.
\end{align*}
Notice that the variable $y_i$ serves as a flag to indicate that the subformula
$\varphi_i$ has been successfully verified for value $x_i=0$. 
Just observe that for any valuation $\nu$ with $\nu(y_i)=0$,
there is a path from $\Qin{i},\nu$ to some $\Qout{i},\nu'$ iff
$\nu_{[x_i=0]}\models\varphi_i$ and $\nu_{[x_i=1]}\models\varphi_i$
Moreover, the existence of such a path implies that $\nu'(y_i)=0$.
An induction on $i$ shows that the given formula is indeed satisfiable if, and only if,
there exists $\nu$ such that $\Qin{1},\nu_0\step{*}\Qout{1},\nu$.

\subsection{An Ackermann upper complexity bound}
Due to our use of Dickson's Lemma in  Lemma~\ref{lem:ef},
we can derive an Ackermannian upper bound for EF model checking using the
approach of Schmitz et.al.~\cite{FFSS2011}.
We show how to to bound the size of our representation of
${\it Pre}^*(S)$
in terms of fast-growing functions. This implies that the space required
by the procedure of Theorem~\ref{thm:EF-dec} can be bounded by
an Ackermannian function.

\newcommand{\F}[1]{F_{#1}}
\newcommand{\FS}[1]{\mathbf{F}_{#1}}
\newcommand{\FR}[1]{\mathcal{F}_{#1}}

The family of \emph{fast-growing functions} $F_n:\N\to\N$
is inductively defined as follows for all $x,k\in\N$.
\begin{align*}
    \F{0}(x) = x+1 \qquad\mbox{and}\qquad
    \F{k+1}(x) = \F{k}^{x+1}(x).
\end{align*}
A variant of the Ackermann function is $\F{\omega}:\N\to\N$,
defined as $\F{\omega}(x)=\F{x}(x)$.

Consider $d$-dimensional tuples of natural numbers with the pointwise ordering $\le$.
A sequence $x_0 x_1\dots x_l  \in (\N^d)^*$ of tuples is called \emph{good} if there exist indices $0\le i<j\le l$
such that $x_i\le x_j$ and \emph{bad} otherwise. I.e., a bad sequence is an
antichain w.r.t.\ the ordering on the tuples.
By Dickson's Lemma, every bad sequence is finite, but there exist bad sequences of
arbitrary length, because there is no assumption on the increase in
one dimension if another dimension decreases.
The \emph{norm} of $x\in\N^d$ is $\norminf{x}=\max\{x(i)\;|\;0\le i\le d\}$.
The sequence is \emph{$t$-controlled} by a function $f:\N\to\N$
if $\norminf{x_i} < f(i+t)$ for every index $0\le i\le l$.
\newcommand{\BS}[3]{L_{#1,#2}(#3)}
Let $\BS{d}{f}{t}$ denote the maximal length of a
bad sequence in $\N^d$ that is $t$-controlled by $f$.

Schmitz et.al.~\cite{FFSS2011},
show how to bound such controlled bad sequences in terms of
fast-growing functions. It follows from their work that
for every $d\ge 1$ and $c,k,x\in\N$,
\begin{equation}
\label{BS-bound}
\BS{d}{\F{k}^c}{t}
\le
\F{k+d-1}^{(c + d +2)^d}(t).
\end{equation}

We are now ready to bound the size of computed representations of
${\it Pre}^{*}(S)$ for a given MG-definable sets $S$.
Fix a GCS with variables $\Var$, constants $\Const$ and
$\delta$-many transitional gap constraints.
For the sake of readability we assume an unlabeled GCS; the bounds we provide
directly apply for the labeled case as well.
Recall that a satisfiable monotonicity graph $M$ has the property that
$M(x,y)<\infty$ for all $x,y\in\Var\cup\Const$.
We identify such a graph with the vector $v_M\in\N^d$ of
dimension
$d=|(\Var\cup\Const)^2|$, where the component for the pair $(x,y)$ has value
$0$ if $M(x,y)=-\infty$ and $n + M(x,y) +1$ otherwise.
In particular, notice that $\norminf{M}$ is bounded by
$n+1+\max\{M(x,y)\mid x,y\in\Var\cup\Const\}$.

Let $S$ be a MG-definable set of valuations represented by a single monotonicity
graph 
and
consider a branch of the tree constructed in the proof of Lemma~\ref{lem:ef}.
It provides a bad sequence $M_0M_1\dots M_l$ where for each $0<i$,
the graph $M_{i}$ 
is the result of composing its predecessor $M_{i-1}$ with one of the transitional monotonicity graphs $G$ of the
system. 
Wlog., assume that all $M_i$ are satisfiable, because otherwise it (and with it
all $M_j$ for $i\le j\le l$) represents the empty set and does not
contribute to ${\it Pre^*}(S)$.
By definition of compositions $G\circ M$ (see Definition~\ref{def:mg-ops})
we therefore get
\newcommand{\maxG}{c}
$\norminf{M_{i}} \le \norminf{M_{i-1}} +\maxG$,
for every $0< i\le l$,
where $\maxG$ is the maximal constant
in the system.
Consequently, the branch is $\norminf{M_0}$-controlled by
$f:x\mapsto x\maxG$.
Since $f$ is dominated by $\F{1}^\maxG$, 
equation~\eqref{BS-bound} provides the bound
\begin{equation}
\label{eq:}
l
~\le~ \BS{d}{f}{t}
~\le~ \BS{d}{\F{1}^\maxG}{t}
~\le~ \F{d}^{(\maxG + d +2)^d}(t)
\end{equation}
on the length of the branch,
where $t=\norminf{M_0}$.
If we instead let $t = \max\{\norminf{M_0}, (\maxG + d +2)^d + \delta \}$,
then we can bound $l$ by
$\F{d}^{(\maxG + d +2)^d}(t) \le \F{d}^{t+1}(t) = \F{d+1}(t)$.
In particular, this means that the norm $\norminf{M_l}$ is bounded by
$\maxG \cdot \F{d+1}(t)\le \F{d+2}(t)$.
Moreover, since $\F{2}^n(x) = x^n+x$, the total number of nodes in the tree is bounded
by
$\delta^l \le \F{2} \F{d+1}(t) \le \F{d+3}(t)$.
We have shown the following lemma.
\begin{lemma}
    \label{lem:prestar-rep}
    \def\B{\F{d+2}(t)}
    Let $S$ be a set of valuations represented by $m$ monotonicity graphs
    and let $t\ge (\maxG + d +2)^d + \delta$ be an upper bound on their norms.
    Then ${\it Pre}^*(S)$ can be effectively represented by
    no more than $m \cdot \F{d+3}(t)$ graphs with norm at most $\B$.
\end{lemma}

The claim of the next lemma directly follows from
Definitions~\ref{def:MGs}, \ref{def:mg-ops}, Lemma~\ref{lem:neg}
as well as the definition of the vector $v_M$ representing the MG $M$.
Notice that complementing a single MG $M$
results in at most $d$ graphs of degree $<\norminf{v_M}$.
Each of them therefore corresponds to a vector with norm
bounded by $2\norminf{v_M} + 1 = \F{1}(\norminf{v_M})$.
\begin{lemma}
    \label{lem:easy-reps}
    Let $S,S'$ be sets of valuations, each
    represented by a set of $m$ monotonicity graphs of
    norm at most $t$.
    Then,
    \begin{enumerate}
      \item $S\cup S'$ can be effectively represented by $m+m$ graphs
          of norm at most $t$,
      \item $S\cap S'$ can be effectively represented by $m\cdot m$
          MG of norm at most $t$,
      \item $\Pre{a}{S}$ can be effectively represented by $m$ graphs
          of norm at most $t+c$, where $c$ is the maximal absolute value
          of any constant in the system,
      \item $\Val \setminus  S$ can be effectively represented by $d^m$ graphs
          of norm at most $\F{1}(t)$.
    \end{enumerate}
\end{lemma}
\begin{proposition}
    Let $(\Var,\Const,\Act,\Delta,\lambda)$ be a gap-order constraint system
    and let $d=(|\Var| +|\Const|)^2$,  $\delta=|\Delta|$ and $c=\max\{|x| : x\in\Const\}$.
    For every EF-formula $\varphi$ of nesting depth $k$, one can effectively
    compute a representation of the set $\DEN{\varphi}$,
    in space $\F{d+4}((c+d+2)^d+\delta + k)$.
\end{proposition}
\begin{proof}
    Let $t=(c + d+2)^d+\delta)$.
    We show by induction on the nesting depth $k$ of subformulae
    that
    $\DEN{\varphi}$ can be represented by at most
    $\F{d+3}^k(t)$ monotonicity graphs with norms bounded by
    $\F{d+2}^k(t)$.
    
    For the base case, observe that atomic propositions are either $\varphi={\it true}$
    or stated as single gap constraint $\varphi=\Con{C}$.
    Either way, $\DEN{\varphi}$ can be expressed as single monotonicity graph
    $M_\varphi$ with norm $\norminf{M_\varphi}\le c\le \F{d+2}^0((c +
    d+2)^d+\delta)$.

    For the induction step, we assume that 
    the claim is true for all formulae of height $i$ and consider
    a formula $\varphi$ of height $i+1$. If the principal connector of $\varphi$
    is $\land,\lor,\neg$ or $\exnext{a}$ for some action $a$,
    then the claim follows by Lemma~\ref{lem:easy-reps}.
    To see this, just notice that for all $m\in\N$,
    $\F{d+2}(m)\ge \F{1}(m)= 2m+1$ and
    $\F{d+3}(m)\ge \F{3}(m)\ge m^m$.
    If $\varphi$ is of the form ${\it EF}\phi$, then,
    by induction hypothesis,
    $\DEN{\phi}$ can be effectively represented by
    no more than $\F{d+3}^i(t)$ graphs with norm $\le \F{2}^i(t)$.
    Lemma~\ref{lem:prestar-rep} thus implies that
    $\DEN{\varphi}$ is representable by $\F{d+3}^{i+1}(t)$ graphs,
    each with norm at most $\F{d+2}^{i+1}(t)$ as required.

    The claim now follows from the observation that the total
    space required for the above representation is
    $d\cdot \log \F{d+2}^k(t) \cdot \F{d+3}^k(t) \le \F{d+4}(t+k)$.
\end{proof}

    \section{Applications}
    \label{sec:app}
We consider labeled transition systems induced by GCS.
In a weak semantics, one abstracts from non-observable actions
modeled by a dedicated action $\tau\in \Act$.
The \emph{weak step} relation $\wstep{}$ is defined by
\[
\begin{array}{lcl}
\wstep{\tau} & = & \Step{\tau}{*}{} \\
\wstep{a}    & = & \Step{\tau}{*}{}\cdot\step{a}\cdot\Step{\tau}{*}{}, \quad\mbox{for $a\neq\tau$}
\end{array}
\]

Bisimulation and weak bisimulation
\cite{Par1981,Mil1989}
are semantic equivalences in van Glabbeek's
linear time -- branching time spectrum \cite{Gla2001}, which are used to compare the behavior of processes.
Their standard co-inductive definition relative to a given LTS is as follows.
\begin{definition}
  A binary relation $R\subseteq V^2$ on the states of a labeled transition
  system is a \emph{bisimulation} if $sRt$ implies that 
  \begin{enumerate}
    \item for all $s\step{a}s'$ there is a $t'$ such that $t\step{a}t'$ and $s'Rt'$, and
    \item for all $t\step{a}t'$ there is a $s'$ such that $s\step{a}s'$ and $s'Rt'$. 
  \end{enumerate}
  Similarly, $R$ is a \emph{weak bisimulation} if in both conditions above, $\step{}$ is replaced by
  $\wstep{}$.
  (Weak) bisimulations are closed under union, so there exist
  unique maximal bisimulation $\mbsim$ and weak bisimulation $\mwbsim$ relations, which
  are equivalences on $V$.

  By the maximal (weak) bisimulation between two LTS with state sets $S$ and $T$
  we mean the maximal (weak) bisimulation in their union projected into $(S\x T)\cup (T\x S)$.
\end{definition}
The \emph{Equivalence Checking Problem} is the following decision problem.

\vspace{0.3cm}
\begin{tabular}{ll}
  {\sc Input:}  &Given LTS $T_1=(V_1,\Act,\step{})$ and $T_2=(V_2,\Act,\step{})$,\\
                &states $s\in V_1$ and $t\in V_2$ and an equivalence $R$.\\
  {\sc Question:} &$s R t$?
\end{tabular}
\vspace{0.3cm}

\noindent In particular, we are interested in checking strong and weak bisimulation between
LTS induced by GCS and finite systems.
Note that the decidability of weak bisimulation implies the decidability of the corresponding strong
bisimulation because $\mbsim$ and $\mwbsim$ coincide for LTS without $\tau$ labels.

Finite systems admit \emph{characteristic formulae} up to
weak bisimulation in \EF\  (see e.g.~\cite{KJ2006,JKM1998}).
\begin{theorem}
    \label{thm:char-formulae}
    Let $T_1=(V_1,\Act,\step{})$ be an LTS with finite state set $V_1$
    and $T_2=(V_2,\Act,\step{})$ be an arbitrary LTS.
    For every state $s\in V_1$ one can construct an \EF-formula $\psi_s$
    such that $t\mwbsim s \iff t\models \psi_s$ for all states $t\in V_2$.
\end{theorem}

The following is a direct consequence of Theorems~\ref{thm:char-formulae} and \ref{thm:EF-dec}.

\begin{theorem}
\label{thm:wbsim-ecp}
    For every GCS $\G=(Var,Const,\Act,\Delta,\lambda)$
    and every LTS $T=(V,\Act,\step{})$ with finite state set $V,$
    the maximal weak bisimulation $\mwbsim$
    between $T_{\G}$ and $T$ is effectively gap definable.
\end{theorem}
\begin{proof}
    By Theorem~\ref{thm:char-formulae} we can compute, for every state $s$ of $T,$
    a characteristic formula $\psi_s$ that characterizes the set of
    valuations $\{\nu\ |\ \nu\mwbsim s\} = \DEN{\psi_s}$.
    By Theorem~\ref{thm:EF-dec}, these sets are MG- and thus gap definable.
    Since the class of gap definable sets is effectively closed under
    finite union
    and $\mwbsim\ = \bigcup_{s\in V}\DEN{\psi_s}$, the result follows.
    \qed
\end{proof}

Considering that gap formulae are particular formulae of Presburger arithmetic,
we know that gap definable sets have a decidable membership problem.
Theorem~\ref{thm:wbsim-ecp} thus implies the decidability of
equivalence checking between GCS processes and finite systems w.r.t.\ strong and weak bisimulation.

    \section{Conclusion and Open Questions}
We have shown that model checking gap-order constraint systems with
the logic \EG\ is undecidable while the problem remains decidable
for the logic \EF.
An immediate consequence of the latter result is the
decidability of
strong and weak bisimulation checking between GCS and finite systems.

The decidability of \EF\ model checking 
is shown by using finite sets of monotonicity graphs
or equivalently, gap formulae to represent intermediate results in a bottom-up
evaluation.
This works because the class of arbitrary gap definable sets
is effectively closed under union and complements
and for a gap definable set $S$ and a GCS $\G$, $\Pre{}{S}$ and ${\it Pre}^*(S)$ are
effectively gap definable.

Our decidability result relies on a well-quasi-ordering argument
to ensure termination of the fixpoint computation for ${\it Pre}^*(S)$,
which does not yield any strong upper complexity bound.
So far, there is only an Ackermannian upper bound and a \PSPACE\ lower bound. 

Interesting open questions include
determining the exact complexity of model checking GCS with respect to
\EF.
We also plan to investigate the decidability and complexity of checking behavioral equivalences
like strong and weak bisimulation between two GCS processes,
as well as checking (weak) simulation preorder and trace inclusion.

    \newpage
    \bibliographystyle{plain}
    \bibliography{bibliography}

\begin{thebibliography}{10}

\bibitem{AD2006}
P~A Abdulla and G~Delzanno.
\newblock Constrained multiset rewriting.
\newblock In {\em Proc. AVIS’06, 5th int. workshop on Automated Verification
  of InfiniteState Systems}. 2006.

\bibitem{BR2000}
Thomas Ball and Sriram~K. Rajamani.
\newblock Boolean programs: A model and process for software analysis.
\newblock Technical Report MSR-TR-2000-14, Microsoft Research, February 2000.

\bibitem{Ben2009}
Amir~M. Ben-Amram.
\newblock Size-change termination, monotonicity constraints and ranking
  functions.
\newblock In Ahmed Bouajjani and Oded Maler, editors, {\em Computer Aided
  Verification}, volume 5643 of {\em LNCS}, pages 109--123. Springer, 2009.

\bibitem{Boz2012}
Laura Bozzelli.
\newblock Strong termination for gap-order constraint abstractions of counter
  systems.
\newblock In {\em LATA}, volume 7183 of {\em LNCS}, pages 155--168. 2012.

\bibitem{BP2012}
Laura Bozzelli and Sophie Pinchinat.
\newblock Verification of gap-order constraint abstractions of counter systems.
\newblock In {\em VMCAI}, volume 7148 of {\em LNCS}, pages 88--103. Springer,
  2012.

\bibitem{Cer1994}
Karlis Cerans.
\newblock Deciding properties of integral relational automata.
\newblock In {\em ICALP}, volume 820 of {\em LNCS}, pages 35--46. 1994.

\bibitem{CS1999}
Stephen Cook and Michael Soltys.
\newblock Boolean programs and quantified propositional proof systems.
\newblock {\em Bulletin of the Section of Logic}, 1999.

\bibitem{DD2007}
St{\'e}phane Demri and Deepak D'Souza.
\newblock An automata-theoretic approach to constraint ltl.
\newblock {\em Inf. Comput.}, 205(3):380--415, 2007.

\bibitem{Esp1997}
Javier Esparza.
\newblock Decidability of model checking for infinite-state concurrent systems.
\newblock {\em Acta Informatica}, 34:85--107, 1997.

\bibitem{FFSS2011}
Diego Figueira, Santiago Figueira, Sylvain Schmitz, and Philippe Schnoebelen.
\newblock Ackermannian and primitive-recursive bounds with dickson's lemma.
\newblock In {\em LICS}, pages 269--278, 2011.

\bibitem{FR1996}
Laurent Fribourg and Julian Richardson.
\newblock Symbolic verification with gap-order constraints.
\newblock In {\em LOPSTR}, volume 1207 of {\em LNCS}, pages 20--37. 1996.

\bibitem{Gla2001}
R.J.~van Glabbeek.
\newblock The linear time -- branching time spectrum {I}; the semantics of
  concrete, sequential processes.
\newblock In J.A. Bergstra, A.~Ponse, and S.A. Smolka, editors, {\em Handbook
  of Process Algebra}, chapter~1, pages 3--99. Elsevier, 2001.

\bibitem{JKM1998}
Petr Jan\v{c}ar, Anton\'{\i}n Ku\v{c}era, and Richard Mayr.
\newblock Deciding bisimulation-like equivalences with finite-state processes.
\newblock In {\em ICALP}, volume 1443 of {\em LNCS}, pages 200--211. 1998.

\bibitem{KJ2006}
Anton\'{\i}n Ku\v{c}era and Petr Jan\v{c}ar.
\newblock Equivalence-checking on infinite-state systems: Techniques and
  results.
\newblock {\em TPLP}, 6(3):227--264, 2006.

\bibitem{May1984}
Ernst~W. Mayr.
\newblock An algorithm for the general petri net reachability problem.
\newblock {\em SIAM J. Comput.}, 13(3):441--460, 1984.

\bibitem{Mil1989}
Robin Milner.
\newblock {\em Communication and concurrency}.
\newblock PHI Series in computer science. Prentice Hall, 1989.

\bibitem{Min1967}
Marvin~L. Minsky.
\newblock {\em Computation: finite and infinite machines}.
\newblock Prentice-Hall, Inc., Upper Saddle River, NJ, USA, 1967.

\bibitem{Par1981}
David Park.
\newblock Concurrency and automata on infinite sequences.
\newblock In {\em TCS}, volume 104 of {\em LNCS}, pages 167--183. 1981.

\bibitem{Rev1993}
Peter~Z. Revesz.
\newblock A closed-form evaluation for datalog queries with integer (gap)-order
  constraints.
\newblock {\em TCS}, 116(1{\&}2):117--149, 1993.

\end{thebibliography}

\end{document}